%% file: M2M_Journal1_v1.7.tex
\documentclass[final]{IEEEtran}
\usepackage{amsthm,amssymb,graphicx,multirow,amsmath,color,amsfonts}
\usepackage[update,prepend]{epstopdf}
\usepackage[noadjust]{cite}
\usepackage[latin1]{inputenc}

\include{notation}
\allowdisplaybreaks 

\begin{document}
\graphicspath{{./Figures/}}
\title{Power-Efficient System Design for Cellular-Based Machine-to-Machine Communications}
\author{Harpreet S. Dhillon, Howard C. Huang, Harish Viswanathan and Reinaldo A. Valenzuela
\thanks{
H. S. Dhillon is with WNCG, the University of Texas at Austin, USA. Email: dhillon@utexas.edu. H. C. Huang, H. Viswanathan and R. A. Valenzuela are with Bell Labs, Alcatel-Lucent, NJ. A part of this paper was presented at IEEE Globecom Workshops 2012~\cite{DhiHuaC2012}. \hfill
Manuscript last revised: \today.}
}

\maketitle
\begin{abstract}
The growing popularity of Machine-to-Machine (M2M) communications in cellular networks is driving the need to optimize networks based on the characteristics of M2M, which are significantly different from the requirements that current networks are designed to meet. First, M2M requires large number of short sessions as opposed to small number of long lived sessions required by the human generated traffic. Second, M2M constitutes a number of battery operated devices that are static in locations such as basements and tunnels, and need to transmit at elevated powers compared to the traditional devices. Third, replacing or recharging batteries of such devices may not be feasible. All these differences highlight the importance of a systematic framework to study the power and energy optimal system design in the regime of interest for M2M, which is the main focus of this paper. For a variety of coordinated and uncoordinated transmission strategies, we derive results for the optimal transmit power, energy per bit, and the maximum load supported by the base station, leading to the following design guidelines: (i) frequency division multiple access (FDMA), including equal bandwidth allocation, is sum-power optimal in the asymptotically low spectral efficiency regime, (ii) while FDMA is the best practical strategy overall, uncoordinated code division multiple access (CDMA) is almost as good when the base station is lightly loaded, (iii) the value of optimization within FDMA is in general not significant in the regime of interest for M2M.

\end{abstract}

\section{Introduction}
The widespread coverage of cellular networks makes them an attractive option for handling the growing number of sensing and monitoring devices. Therefore, M2M communications, involving wide area communication of sensor data to an Internet based application, is emerging as an important service over mobile cellular networks~\cite{FocM2003,VodM2010,EriM2011,LieCheJ2011}. It spans multiple vertical industries such as transportation, healthcare, utilities, retail, industrial monitoring, banking, and home automation and includes a variety of applications within each vertical. Projections for growth of M2M communication devices range from 24 billion \cite{GSMAM2012} to 50 billion \cite{EriM2011a} in the next decade with over 2 billion M2M devices expected to directly attach to the cellular network by this time. Given the potential for a significant new revenue stream from M2M data services, the industry is focusing on ensuring that cellular networks can efficiently serve the needs of M2M communications.  

\subsection{Motivation and Related Work}
M2M devices in some verticals are deployed in locations that are not frequented by people. For example, vending machines and  water, gas or other meters, are typically deployed in basements of buildings, water monitoring systems are deployed underground, and some traffic monitoring systems may be deployed in the tunnels. Such machines will typically have devices that communicate with their controlling applications or servers over wireless networks. For such locations, radio signals have to be substantially stronger compared to what is required for traditional service. Hence, device transmit power required for communication becomes a critical issue. In addition, if the devices are battery operated total transmission energy from the device is also an important consideration.

For M2M, the goal of minimizing mobile transmit power is aided by the nature of M2M traffic. In contrast to traditional consumer traffic, M2M typically involves a large number of short payload transactions, as shown in \cite{VisMunJ2012,ShaJiC2012}, where traffic models for various vertical applications are characterized. For example, a fleet management application can involve transmission of location every 20 seconds by each vehicle to the central application server with each transaction involving less than 500 bytes \cite{JouAttC2011}. Similarly, reporting of health data such as blood pressure or heart rate by medical devices involves payloads less than 200 bytes \cite{IEEEM2008}.



Mobile cellular networks, including the fourth generation Long Term Evolution (LTE)~\cite{GhoZhaB2010} are neither designed with link budget requirements of M2M devices, nor optimized for M2M traffic pattern.
The system design is optimized to maximize spectral efficiency and minimize latency. Thus mobiles transmit in short duration bursts at high power levels that maximize the total sector throughput. Mobile transmit power levels are primarily dictated by maximum transmit power limit of mobiles and out-of-cell interference considerations. Battery life is not a primary concern for communications since the dominant power consumption on human devices are driven by displays and complex application processing. Furthermore, users will recharge their devices as required.

Clearly, M2M communications impose new requirements on cellular networks that demand rethinking some of the design principles~\cite{3GPPM2010,ZheHuJ2012}. First, for short transactions it may be advantageous to transmit payload in the random access request itself instead of establishing dedicated bearers~\cite{CheWanC2010}. This depends on the size of the payload, the overheads involved and the level of base station loading. Second, the optimization criterion for resource allocation and transmit power levels is average transmit power or energy consumed to transmit a given payload.  This is because for M2M devices battery life is an important concern, and communications consumes a significant fraction of the battery energy, especially if the devices have to transmit at elevated powers due to their adverse locations. Motivated by these differences, several modifications in the current communication protocols to reduce signaling overhead~\cite{MarKirC2005,CheYanC2009,CheWanC2010} and power consumption~\cite{ChaCheC2011} have been proposed in the literature. Furthermore, the idea of cooperative design where several devices are clustered together with a possibility of a controller acting as a common link between a cellular base station and the devices is investigated, e.g., in~\cite{BarHerC2011,TuHoC2011}. The problem of uplink scheduling of M2M devices in LTE networks is studied in \cite{LioAleC2011}, where it is shown that it is better not to form different classes of the devices in order to increase the maximum load that can be served at the base station. Despite these research efforts, there is little understanding of the fundamental tradeoffs in the parameter space of interest in M2M communications, especially from a power and energy optimal design perspective, which is the main focus of this paper.

\subsection{Contributions}

\noindent {\bf Evaluate different multiple access strategies and identify the minimum power strategy.} We study both uncoordinated and coordinated multiple access strategies in this paper. In uncoordinated strategies, the payload is carried in the very first message together with the control information such as the device identity and thus there are no dedicated resources allocated. We consider both FDMA and CDMA random access strategies. The transmit power, number of frequency channels in FDMA and spread factor in CDMA are adjusted based on the average load on the system, which is known at the devices by listening to downlink broadcast signaling from the base station. Interested readers can refer to~\cite{VisJ2009} for a comparison of the two schemes from a throughput perspective. We then consider coordinated strategies in which the resources are explicitly scheduled to active devices. The base station determines the transmission time duration, bandwidth and power and indicates this to the devices. Here we consider successive interference cancellation (SIC), FDMA and time division multiple access (TDMA) strategies and focus on determining the optimal bandwidth/time and transmit power settings for minimizing the total power and/or energy consumed to transmit a fixed size payload. Comparison of the uncoordinated and coordinated strategies shows, not surprisingly, that the coordinated strategies outperform the uncoordinated strategies for heavy loading whereas the performance is comparable for light loads. Hence for low loads, an uncoordinated strategy may be preferred when taking into account downlink overhead required to inform the devices of the allocated resources.

\noindent {\bf Bound the gap between optimal and equal-resource allocation in coordinated orthogonal transmission.} Minimum energy and/or power scheduling over time or frequency is known to be a convex optimization problem and various efficient algorithms have been proposed in the literature~\cite{PraBiyC2001,ElNaiC2002,KesKodC2003}. We show that, in the cellular setting, it is possible to approximate the optimal resource allocation through simpler one-shot solutions, which leads to closed form expressions for the optimization parameters in certain special cases. Furthermore, we prove that in the parameter space of interest, somewhat surprisingly, simple equal-resource allocations perform within a very small factor of the optimal allocations both in terms of transmit power and energy per bit for both TDMA and FDMA. The bounding factor is derived in closed form for all the cases and is around $1.25$ for both FDMA sum-power and sum-energy and also for TDMA sum-energy. It is around $2$ for TDMA sum-power. This suggests that a simple equal resource allocation algorithm is sufficient to achieve near-optimal transmission in terms of both the power and energy minimization.

\noindent {\bf Asymptotically optimal strategies for coordinated transmission.} Using the closed form bounds derived to bound the gap between the optimal and equal resource allocation, we further show that the FDMA equal resource allocation is both sum-power and sum-energy optimal over the space of FDMA strategies in the limit of asymptotically low spectral efficiency. Since the closed form bound is the same for TDMA energy optimal solution, the result extends in this case as well. Using these bounds and the closed form result derived for the sum-power required in SIC, we also prove that FDMA, including equal resource allocation strategy, is sum-power optimal over general resource allocation strategies (not necessarily orthogonal) in the limit of asymptotically low spectral efficiency.

The rest of the paper is organized as follows. We introduce notation and describe the system model in Section II. We discuss the uncoordinated and coordinated access strategies and the associated system designs in Sections III and IV, respectively. We present the numerical results comparing different strategies in Section V. We present our conclusions in Section VI.

\section{System Model \label{sec:sysmod}}

\begin{figure}[t]
\centering
\includegraphics[width=.6\columnwidth]{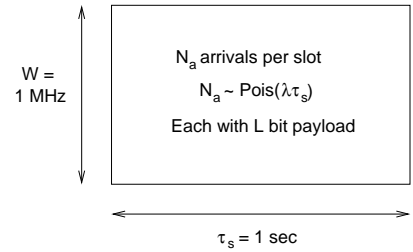}
\caption{Illustration of the time-frequency resource ``slice'' over which multiple users are scheduled.}
\label{fig:ResourceBlock}
\end{figure}

\subsection{System Setup}
In this paper, we consider a single cell consisting of an access point lying at the origin and devices uniformly distributed around it in an annular region with inner and outer radii $r_{i}$ and $r_{0}$, respectively. The non-zero inner radius is assumed to avoid singularity in the path loss model, which is discussed later in this section. We focus on a single cell system with no out-of-cell interference, which is one of the simulation scenarios in 3GPP model~\cite{3GPPM2010}. It should be noted that the out-of-cell interference effectively changes the operating signal-to-interference-plus-noise ratio ($\sinr$), which can be incorporated in the current analysis to study the multi-cell case. In this study, we focus only on the uplink. To characterize the uplink load seen by the base station, we model the arrival process of packets as a Poisson point process with mean $\lambda$ arrivals per second.  
For concreteness, we assume a time slotted system with the slot duration denoted by $\tau_s$. The analysis will be performed on a typical time-frequency resource slice with slot duration $\tau_s$ and bandwidth $W$ as shown in Fig.~\ref{fig:ResourceBlock}. 
We denote the number of packet transmission requests in each block by $N_a \sim \pois(\lambda \tau_s)$. Each packet is assumed to have a payload of $L$ bits.

\subsection{Multiple Access}
Uplink multiple access is primarily enabled by a broadcast or beacon signal that is transmitted by the base station at the start of each time slot. The beacon signal is assumed to carry load information, which is characterized by the average load $\lambda$ seen at the base station. We consider two broad categories of multiple access strategies: i) uncoordinated: the devices transmit data using slotted random access and there is no need to establish dedicated radio bearers, and ii) coordinated: the devices transmit data in a separate scheduled transmission. An example of the system using uncoordinated strategy is the one-stage setup illustrated in Fig.~\ref{fig:1stage}, where both the control information and the data are transmitted as a part of the random access signal. As discussed in the sequel, the transmit power and other transmission parameters are independently adjusted by each mobile device based on the load information, $\lambda$, that is received in the beacon signal. 

\begin{figure}[ht]
\centering
\includegraphics[width=.8\columnwidth]{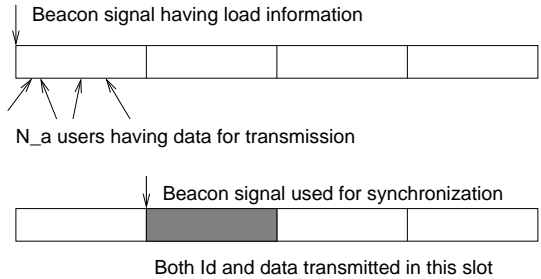}
\caption{One-stage design where both device Id and data are transmitted together in random access stage.}
\label{fig:1stage}
\end{figure}

\begin{figure}[ht]
\centering
\includegraphics[width=.8\columnwidth]{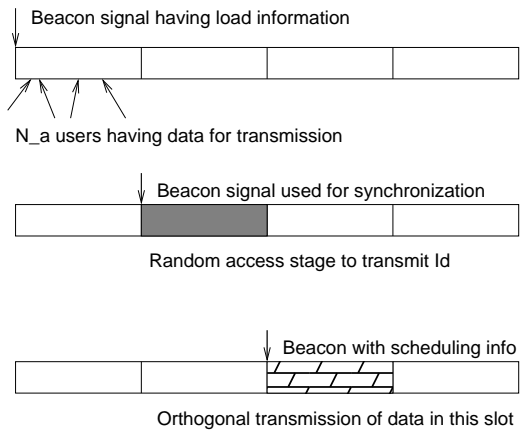}
\caption{Two-stage design where device Id is transmitted in the random access stage to establish connection followed by the scheduled data access stage.}
\label{fig:2stage}
\end{figure}

An example of the system using a coordinated strategy is the the two-stage setup shown in Fig.~\ref{fig:2stage}. In this setup, a dedicated uplink connection is first established by transmitting only the control information in the random access stage. The random access design is the same as the one discussed for the one-stage system. The base station then transmits scheduling information as a part of the beacon signal, which is used by the devices to transmit orthogonally from each other. 

In all the strategies considered in this paper, we assume that all the packet transmissions occur in the current slot and none of them are left for scheduling in a future slot, thereby introducing a notion of a packet deadline explicitly.
Furthermore, we only consider successful packet transmission in a single attempt since the packet failure rate will be small by design. The effective deadline for a single slot design is $\tau_s$ and for two-stage design is $2\tau_s$, counting time from the beginning of the first slot after the packet arrives.

\subsection{Modeling Uplink Channel from Device to Base Station}
The received power at the base station from a device located at distance $r$, assuming transmit power $P_t$, path loss exponent $\gamma$, small-scale fading gain $h$, large scale shadowing gain $\ncalX$ and direction based antenna gain $G$ is:
\begin{align}
P_r = P_t \ncalX h G r^{-\gamma}.
\end{align}
Instead of treating these link budget parameters individually, we model their composite effect by defining the reference signal-to-noise ratio ($\snr$) $\mu$ as the average received $\snr$ from a device transmitting at maximum power $P_{\max}$ over bandwidth $W$ located at cell edge, i.e., at distance $r_0$ from the base station. Therefore, the received $\snr$ $\mu_r$ at the base station from a device located at distance $r$ can be expressed in terms of $\mu$ as:
\begin{align}
\mu_r = \frac{P_t}{P_{\max}} \mu \ncalX h \left(\frac{r}{r_0}\right)^{-\gamma}.
\end{align}
Please note that reference $\snr$ is a function of signal bandwidth because of the scaling of noise power with bandwidth. We will comment more on this when we study coordinated access using TDMA and FDMA in the sequel. Now defining the channel gain $g = \ncalX h \left(\frac{r}{r_0}\right)^{-\gamma}$, we get the following simple expression for the received $\snr$ $\mu_r$:
\begin{align}
\mu_r = \frac{P_t}{P_{\max}} \mu g.
\label{eq:SNR}
\end{align}

In case the information symbols are transmitted over a bandwidth $W_\calN < W$, the reference SNR $\mu_\calN$ in this case can be written in terms of $\mu$ as:
\begin{equation}
\mu_\calN = \mu \frac{W}{W_\calN},
\end{equation}
which is greater than $\mu$ because of the decrease in the effective noise power. For this modified system, the received $\snr$ can be evaluated from \eqref{eq:SNR} by replacing $\mu$ with $\mu_\calN$. This will be helpful in analyzing the multiple access strategies that involve partitioning of frequency resources, e.g., FDMA. Throughout this paper, we assume that the channel gain $g$ is known at the device both in the cases of uncoordinated and coordinated strategies.

For all numerical results, we consider $P_{\max} = 1$ W, $r_0 = 1000$ m, $\gamma = 3$, $\mu=-3$ dB, $W=1$ MHz, $\tau_s=1$ sec and ignore fading and shadowing. A careful reader would observe that the value of $\mu$ is very low for a reference distance of $1000$ m. This is chosen to account for the high penetration losses suffered by the radio signals when the devices are deployed at adverse locations such as basements and tunnels. Since we have chosen $P_{\max} = 1$ W, we will drop it from \eqref{eq:SNR} to reduce it to $\mu_r = P_t \mu g$ with the understanding that the $P_t$ is the normalized by $P_{\max} = 1$ W.
 
\subsection{Preliminaries}
In the slotted system introduced above, an information symbol of payload $L$ bits will be transmitted over bandwidth $W_\calN \leq W$ for time $\tau \leq \tau_s$ depending upon the resources allocated to that packet. To fix the key ideas, we first confine our discussion to the perfectly orthogonal resource allocation using TDMA and FDMA, where one of these inequalities will be a strict inequality due to the partitioning of time-frequency resources.
Assuming capacity achieving codes, $\tau$ and $W_\calN$ are related to the received $\snr$ $\mu_r$ by Shannon's capacity equation as follows:
\begin{equation}
\frac{L} {\tau W_\calN} = \log_2 \left( 1 + \mu_r \right).
\label{eq:Shannon}
\end{equation}
It should be noted that the effect of finite block length can be easily incorporated in the above expression by means of an $\snr$ gap. Interested readers can refer to~\cite{SleJ1963} for more details. Using \eqref{eq:SNR} and \eqref{eq:Shannon}, we can find the minimum transmission time required to transmit $L$ bits over bandwidth $W_\calN$ under the maximum transmit power $P_{\rm max}$ constraint as:
\begin{equation}
\tau \ge \frac{L}{W_\calN \log_2 (1+P_{\rm max} \mu_\calN g)} = \tau_{\rm min}.
\label{eq:tau_min}
\end{equation}
Similarly, the minimum transmission bandwidth required for the transmission of $L$ bits over time $\tau$ is given by the solution of the following equation:
\begin{equation}
\frac{L} {\tau W_{\rm min}} = \log_2 \left( 1 + P_{\rm max} \mu \left( \frac{W}{W_{\rm min}} \right) g  \right).
\label{eq:W_min}
\end{equation}
As will be evident later, \eqref{eq:tau_min} and \eqref{eq:W_min} will be useful in formulating the optimization problems for TDMA and FDMA cases, respectively. These arguments easily extend to the CDMA case and are discussed as a part of the uncoordinated strategies in the next section.

\section{Uncoordinated Transmission}

In uncoordinated transmission, we assume that $L$ bits are transmitted in each transaction with the understanding that $L$ would be small ($\approx 50$ bits) when only control information, such as device identity, is transmitted and would be relatively large ($\approx 1000$ bits) when both control information and data are transmitted together. For device multiplexing, we consider two strategies: i) CDMA random access, and ii) FDMA random access. We start this discussion with the design of CDMA random access, from which the FDMA random access design will follow.

\subsection{CDMA Random Access}
In CDMA random access, we assume that each device selects a $N_c$ length code randomly from the set of $2^{N_c} - 1$ possible binary sequences, where $N_c$ is the design parameter. As a result of this random code selection, the chosen codes will not be perfectly orthogonal in general. Since the transmissions of various devices are synchronized, the cross correlation in two randomly chosen codes is assumed to be $1/N_c$. For fair comparison across all the transmission strategies considered in this paper, we assume that the total bandwidth is $W$ over which the CDMA waveform is transmitted, leading to $W/N_c$ as the effective bandwidth of the information symbols. We further assume that the devices perform uplink power control such that the target $\sinr$ at the base station is $\mu_t$.

\subsubsection{System Design}
As is clear from the setup, the only information devices have about the number of transmission requests $N_a \sim \pois(\lambda \tau_s)$ is the mean load $\lambda \tau_s$. Therefore, we design the system for $(1-\epsilon)^{th}$ percentile of the arrivals, which is denoted by $\bar{N}$, where $\P[N_a > \bar{N}] \leq \epsilon$. This also defines first failure event, i.e., when the actual number of arrivals are greater than $\bar{N}$ and it is no longer possible to achieve the target $\sinr$ while satisfying the maximum power constraint. We will use this later in our discussion to bound the overall failure probability.

After determining $\bar{N}$ from the load information, the next step is to determine the length $N_c$ of the spreading codes. For this, we first define the collision event as follows.
\begin{ndef}[Collision Event $\calA_c$]
Collision is said to occur when more than one device choose the same spreading code.
\end{ndef}
Clearly, collision leads to packet failures since there is no way to differentiate between various devices. For given $\bar{N}$ and $N_c$, the probability of collision is given by the following Lemma.
\begin{lemma}
\label{lem:Pc}
Defining $m (= 2^{N_c}-1)$ as the number of codes and $\bar{N}$ as the number of arrivals, the probability of collision can be expressed as
\begin{equation}
\pc = \P[\mathcal{A}_c] = 1 - \left(1-\frac{1}{m}\right)^{\bar{N}-1}.
\end{equation}
\end{lemma}
\begin{proof}
First note that the probability of a particular code being chosen by $n$ devices is
\begin{align}
\P[n] = {\bar{N} \choose n} \left( \frac{1}{m} \right)^n  \left(1 - \frac{1}{m}  \right)^{\bar{N}-n}.
\end{align}
The result follows from the fact that
\begin{align}
\pc &= 1 - \P[0] - \P[1] = \P[\mathcal{A}_c]\\
 &= 1 - \left(1-\frac{1}{m}\right)^{\bar{N}} - \frac{1}{m} \left(1 - \frac{1}{m}  \right)^{\bar{N}-1}.
\end{align}
\end{proof}

\begin{cor}
The length of the spreading sequence required to restrict the collision probability within a predefined value $\pc$ is
\begin{align}
N_c &= \left\lceil \log_2 \left(1+ \frac{1}{1 - (1-\pc)^{\frac{1}{\bar{N} - 1}}}  \right)  \right\rceil \\
&\approx \left\lceil \log_2 \left(1 + \frac{\bar{N}-1}{\pc}\right)\right\rceil.
\end{align}
\end{cor}
\begin{proof}
The main result follows directly form Lemma~\ref{lem:Pc} and the approximation follows in the low collision probability regime from the fact that
\begin{align}
\lim_{\pc \rightarrow 0}(1-\pc)^{\frac{1}{\bar{N} - 1}} &\stackrel{(a)}{=} 1 + \frac{\ln(1-\pc)} {\bar{N}-1}\\
&\stackrel{(b)}{=} 1 - \frac{\pc}{\bar{N}-1},
\end{align}
where $(a)$ follows form the Taylor series expansion of $a^{x}$ and $(b)$ follows from the fact that $\lim_{x\rightarrow 0} \ln (1 - x) = -x$. An alternate way to directly prove the approximate result is by observing that the collision probability $\pc$ can be tightly upper bounded by
\begin{equation}
\pc \leq \frac{\bar{N}-1}{m}.
\end{equation}
To prove this bound, define $M$ as the random variable denoting the number of unique codes occupied by $\bar{N}-1$ devices. The collision probability can now be expressed as
\begin{equation}
\pc = \frac{E[M]}{m} \leq \frac{\bar{N} - 1}{m},
\end{equation}
which completes the proof.
\end{proof}

For a given $N_c$ and channel gain $g$, the transmit power required to achieve target $\sinr$ $\mu_t$ can be derived in the closed form and is given by the following Lemma. This will be useful when we compare and contrast various access techniques in Section~\ref{sec:numresults}.

\begin{lemma}[CDMA Transmit Power] \label{lem:CDMA_Pt}
The transmit power and the energy per bit required by a mobile to achieve target $\sinr$ $\mu_t = 2^\frac{L N_c}{W \tau_s} - 1$ with channel gain $g$ and code length $N_c$ are
\begin{align}
P_t &=  \left[ \mu g\left(N_c \mu_t^{-1} - (\bar{N}-1) \right) \right]^{-1} \\
E_b &= \frac{\tau_s}{L} \left[ \mu g\left(N_c \mu_t^{-1} - (\bar{N}-1) \right) \right]^{-1}.
\end{align}
\end{lemma}

\begin{IEEEproof}
From Shannon's capacity equation, we have
\begin{equation}
\frac{L N_c}{W \tau_s} = \log_2(1+\mu_t) \Rightarrow \mu_t = 2^\frac{L N_c}{W \tau_s} - 1
\end{equation}
Assuming cross correlation in spreading codes to be $1/N_c$ and the effective symbol bandwidth to be $W/N_c$, the target $\sinr$ $\mu_t$ can be expressed as
\begin{equation}
\mu_t = \frac{N_c}{(\bar{N}-1) + \mu_f^{-1}} \Rightarrow \mu_f = \left[N_c \mu_t^{-1} - (\bar{N}-1)  \right]^{-1},
\end{equation}
where $\mu_f$ is the fixed received $\snr$ at the base station. The required transmit power can now be expressed as
\begin{equation}
\mu_f = P_t\mu g \Rightarrow P_t =  \left[ \mu g\left(N_c \mu_t^{-1} - (\bar{N}-1) \right) \right]^{-1},
\end{equation}
which completes the proof.
\end{IEEEproof}


\subsubsection{Failure Probability}
There are three failure events possible: i) $N_a > \bar{N}$, in which case we have $\mu_r < \mu_t$ and hence all the packets are in error, ii) collision event $\calA_c$ in which two users pick the same random sequence, and iii) $P_t$ derived in Lemma~\ref{lem:CDMA_Pt} is more than $P_{\rm max}$ for a particular device so that the maximum power constraint is violated (device is in outage). Denote this outage probability by $\delta$. As a matter of convention, we assume that the devices in outage transmit at their maximum power. This is justified because it is highly likely (with probability $1-\epsilon$) that the actual number of arrivals are much less than $\bar{N}$ for which the system is designed. This may ensure successful transmission in certain cases even though the transmit power was less than the designed value. In this case, the outage probability $\delta$ defined as $\P[P_t > P_{\rm max}]$ is an upper bound on the actual outage. For the given failure probabilities $\epsilon$, $\pc$ and $\delta$, the overall failure probability $P_f$ is upper bounded by the following Proposition.
\begin{prop} \label{prop:CDMA_Pf}
The overall failure probability in CDMA random access case is
\begin{equation}
P_f \leq \epsilon + (1 - \epsilon) \{\delta + \pc (1-\delta) \},
\end{equation}
where the inequality reduces to equality when $\delta=0$.
\end{prop}
\begin{proof}
For a general failure event $\calA$, the failure probability $P_f = \P[\calA]$ can be expressed as:
\begin{align}
P_f =&\  \P[\mathcal{A} | N_a > \bar{N}] \P[N_a > \bar{N}] + \P[\mathcal{A} | N_a \leq \bar{N}] \P[N_a \leq \bar{N}] \\
=&\ \epsilon + (1-\epsilon)\{\P[\mathcal{A} | N_a \leq \bar{N}, P_t > P_{\rm max}] \P[P_t > P_{\rm max}] + \nonumber \\
&\ \P[\mathcal{A} | N_a \leq \bar{N}, P_t \leq P_{\rm max}] \P[P_t \leq P_{\rm max}] \},
\end{align}
from which the result follows by observing that the failure events corresponding to the first and second probability terms are, respectively, device outage and collision of the chosen random sequences.
\end{proof}
This result shows that given the overall failure probability, the error probabilities $\epsilon$, $\pc$ and $\delta$ can not be independently chosen. Observe that the above upper bound is tight for small values of $\delta$, which leads to the following characterization of the maximum load supported by a base station.
\begin{cor}[Maximum Load]
The maximum arrival rate (load) supported by the base station for given failure probability $p$ and given outage probability $\delta$ is
\begin{equation}
\lambda_{\rm max} = \max \left\{\lambda: \P \left( P_t (\lambda) > P_{\rm max}  \right) \leq \delta, P_f \leq p \right\}.
\end{equation}
\end{cor}

\begin{example}[CDMA Maximum Load]
\label{eg:CDMA}
For general system parameters $L=1000$ bits, $W=1$ MHz, $\mu = -3$ dB, $\tau_s = 1$ sec, and CDMA specific parameters $P_f = .05$, $\epsilon = .01$, $\delta=0$,  and no fading, the maximum load that a base station can handle under CDMA random access is $\lambda \approx 1350$ arrivals per second.
\end{example}

\subsection{FDMA Random Access}
FDMA random access design follows on the same lines as that of CDMA random access, with the only difference that the devices now choose one of the $N_f$ orthogonal channels and when the two devices choose different channels, there is no interference. Using the collision probability result derived in Lemma~\ref{lem:Pc}, the number of orthogonal channels required in this case can be expressed as:
\begin{align}
N_f = \frac{1}{1 - (1-\pc)^{\frac{1}{\bar{N} - 1}}},
\end{align}
from which the transmit power can be expressed as:
\begin{align}
P_t = \frac{2^{\frac{L N_f}{W \tau_s}} - 1}{\mu g N_f}.
\end{align}
For this new transmit power expression, the maximum load supported by a base station can be defined in the same way as done for the CDMA case.
\begin{example}[FDMA Maximum Load]
\label{eg:FDMA_RandAccess}
For the same system parameters as that of Example~\ref{eg:CDMA}, the maximum load that a base station can handle under FDMA random access is $\lambda \approx 160$ arrivals per second, which is order of magnitude lower than that of CDMA case.
\end{example}

\section{Coordinated Transmissions}

In this section, we discuss the design of coordinated strategies. We assume that the dedicated uplink connections are already established in the random access stage as discussed in the previous section. As a result, the base station has complete information about the $N_a$ arrivals, which are to be coordinated. As is clear from \eqref{eq:tau_min} and \eqref{eq:W_min} derived in Section~\ref{sec:sysmod}, there is a minimum transmission time $\tau_{\rm min}$ and minimum bandwidth $W_{\rm \min}$ required for the successful transmission of a given payload $L$, which appears due to the constraint on maximum power at which a device can transmit. Due to these constraints on the minimum resources required, there is clearly a fundamental limit on the number of packets that can be scheduled in a given time-frequency resource block. This limit defines the maximum load a base station can handle and is dependent upon the resource partitioning strategy being employed. Since the devices farther out near the cell edge or in deep fade require more resources than the others, we assume that the base station deliberately drops a small fraction of these arrivals to increase the maximum load it can handle and to reduce the overall average transmit power or energy~\cite{NeeJ2009}.

For concreteness, we assume that the base station always drops $\delta_1$ fraction of the arrivals, while ensuring that the total system outage is always less than $\delta$. We will comment on the relationship between $\delta_1$ and $\delta$ later in this section. Denote the number of devices actually served in each slot by $K \leq (1-\delta_1)N_a$. It should be noted that the dropped arrivals are the ones having smallest channel gains. For resource partitioning, we consider two approaches: i) TDMA -- splitting time slot $\tau_s$ into $K$ parts, and ii) FDMA -- splitting bandwidth $W$ into $K$ parts, where the goal in both the approaches is to minimize total transmit power or energy. 
As shown later in this section, the optimization problems to find power and energy optimal schedules are convex and hence can be solved efficiently using known algorithms, such as the MoveRight algorithm~\cite{ElNaiC2002}. However, we show that the numerical optimization is not necessary since a near-optimal tractable solution can be found where the resources allocated to each device solely depend upon its channel gain and are independent of the channel gains of other devices. More interestingly, we analytically show that both the average transmit power and energy under equal resource allocation is within a small factor of the optimal values in both TDMA and FDMA. Furthermore, we also compare the results with the power optimal multiuser detection strategy, i.e., SIC, and show that FDMA (including simple equal bandwidth allocation) is power optimal in the limit of asymptotically low spectral efficiency.

In all the coordinated strategies discussed above, once the base station decides the resource allocation it is relayed to the devices as a part of the beacon signal, which is then used by the devices to transmit over the allocated resources. We first derive the power required in SIC, which is global optimal in terms of total power minimization and will serve as a benchmark for the orthogonal transmission strategies discussed later in this section.

\subsection{Power Optimal Coordinated Strategy}
For $K$ users, assuming the channel gains are ordered in the increasing order, i.e., $g_1 \leq g_2 \leq \ldots g_K$, the optimal strategy is the weakest-last interference cancellation strategy, in which the user with the strongest channel gain is decoded first assuming interference from the other users as noise~\cite{TseVisB2005}. This signal is then cancelled from the received signal while decoding the user with the second best channel gain. The process is repeated until the last user is decoded. Therefore, the user with the strongest channel gain sees interference from all the other users, whereas the user with the weakest channel gain does not see any interference. The total transmit power under this strategy is given by the following theorem.
\begin{theorem}
For $K$ users with $g_1 \leq g_2 \leq \ldots g_K$, the total transmit  power for the weakest-last successive interference cancellation strategy is:
\begin{align}
P = \frac{2^{ \frac{L}{W \tau_s} } -1 }{\mu}   \sum_{k=1}^K \frac{2^{ \frac{(k-1)L}{W \tau_s} } }{g_k}.
\label{eq:SIC_sumP}
\end{align}
\end{theorem}
\begin{proof}
Under this particular decoding strategy, the Shannon's capacity expression of the $k^{th}$ user is:
\begin{align}
\frac{L}{W\tau_s} = \log_2 \left( 1 + \frac{P_k \mu g_k}{1 + \sum_{i=1}^{k-1} P_i \mu g_i}  \right),
\end{align}
from which we can derive the following relationship in the transmit powers:
\begin{align}
P_k \mu g_k = \left(2^{\frac{L}{W\tau_s}} - 1 \right)\left(1 + \sum_{i=1}^{k-1} P_i \mu g_i   \right).
\end{align}
Starting from $k=1$ and solving for the transmit powers $P_k$ by substutiting all the preceding values in the above expression leads to the following closed form expression for $P_k$:
\begin{align}
P_k = \frac{2^{ \frac{(k-1)L}{W \tau_s} } \left(2^{ \frac{L}{W \tau_s} } -1 \right)}{g_k \mu},
\end{align}
which completes the proof.
\end{proof}

We now discuss the TDMA system design in detail and then show that the results for FDMA case directly follow.

\subsection{TDMA System Design}
We first find the optimal power or energy partition of $\tau_s = \{\tau_1, \tau_2, \ldots, \tau_{K}\}$ such that the $i^{th}$ device transmits its $L$ information bits over bandwidth $W$ in time $\tau_i$. Using Shannon's capacity expression, the total transmit power required under this allocation is:
\begin{align}
P =  \sum\limits_{i=1}^{K} \frac{2^{\frac{L}{W \tau_i}}-1}{\mu g_i}.
\end{align}
Similarly the total energy per bit can be expressed as:
\begin{align}
E_b = \sum\limits_{i=1}^{K} \frac{\tau_i}{L}\frac{\left(2^{\frac{L}{W \tau_i}}-1\right)}{\mu g_i}.
\end{align}
As discussed above, the maximum power constraints leads to the following constraint on the transmission time:
\begin{align}
\tau_i \ge \frac{L}{W \log_2 (1+P_{\rm max} \mu g_i)}.
\end{align}
The optimization problem to find the power (or energy) optimal schedule can now be formulated as follows:
\begin{align}
\min_{\{\tau_i\}}\ \ \ & \sum\limits_{i=1}^{K} u(\tau_i) \nonumber \\
s.t.\ \ \ & \sum\limits_{i=1}^{K} \tau_i \leq \tau_s \nonumber \\
& \tau_i \geq \frac{L}{W \log_2 (1+P_{\rm max} \mu g_i)} (= \tau_{{\rm min}_i}) \nonumber \\
& 1 \leq i \leq K, \nonumber
\end{align}
where the function $u(\tau_i)$ is chosen appropriately depending upon whether the goal is to minimize total power or energy. Regardless of this choice, the optimization problem remains convex and hence can be efficiently solved using known algorithms~\cite{ElNaiC2002}. Clearly the power and energy optimal allocations are in general different for TDMA.

\begin{remark}[Feasibility and Maximum Load]
\label{rem:TDMAfeasibility} The above optimization problem is feasible if the constraint $\sum_{i=1}^K \tau_{{\rm min}_i} \leq \tau_s $ is satisfied. Using this constraint, an approximation on the maximum load can be derived as follows:
\begin{align}
\frac{\sum_{i=1}^K \tau_{{\rm min}_i}}{K} \leq \frac{\tau_s}{K}
\Rightarrow
K \leq \frac{\tau_s}{\sum_{i=1}^K \tau_{{\rm min}_i}/K}
\stackrel{(a)}{\approx} \frac{\tau_s}{\E[\tau_{\rm min}]},
\end{align}
where $(a)$ follows from two approximations: i) Strong Law of Large Numbers (SLLN) holds and the average converges to the mean of the random variables, ii) the mean is $\E[\tau_{\rm min}]$, which is not exact because $\delta N_a$ arrivals with smallest channel gains have been removed and hence order statistics should be used to compute the mean of the remaining random variables. Nevertheless, this approximation is tight when $\delta \rightarrow 0$, which is the regime of interest in this system design. More formally, the maximum load can be defined in terms of a given outage constraint as follows:
\begin{align}
\lambda_{\rm max} = \max \left\{ \lambda: \P\left[ N_a (1-\delta_1) \geq K_{\rm max}     \right] \leq \epsilon_1 \right\},
\end{align}
where $K_{\rm max}$ is defined by ordering the arrivals in terms of decreasing channel gains as follows:
\begin{align}
K_{\rm max} = \max_K \sum_{i=1}^K \tau_i \leq \tau_s.
\end{align}
The parameters $\delta_1$ and $\epsilon_1$ should be chosen such that the total system outage is less than $\delta$. This is made precise in the following Proposition.
\end{remark}

\begin{figure}[t]
\centering
\includegraphics[scale=.45]{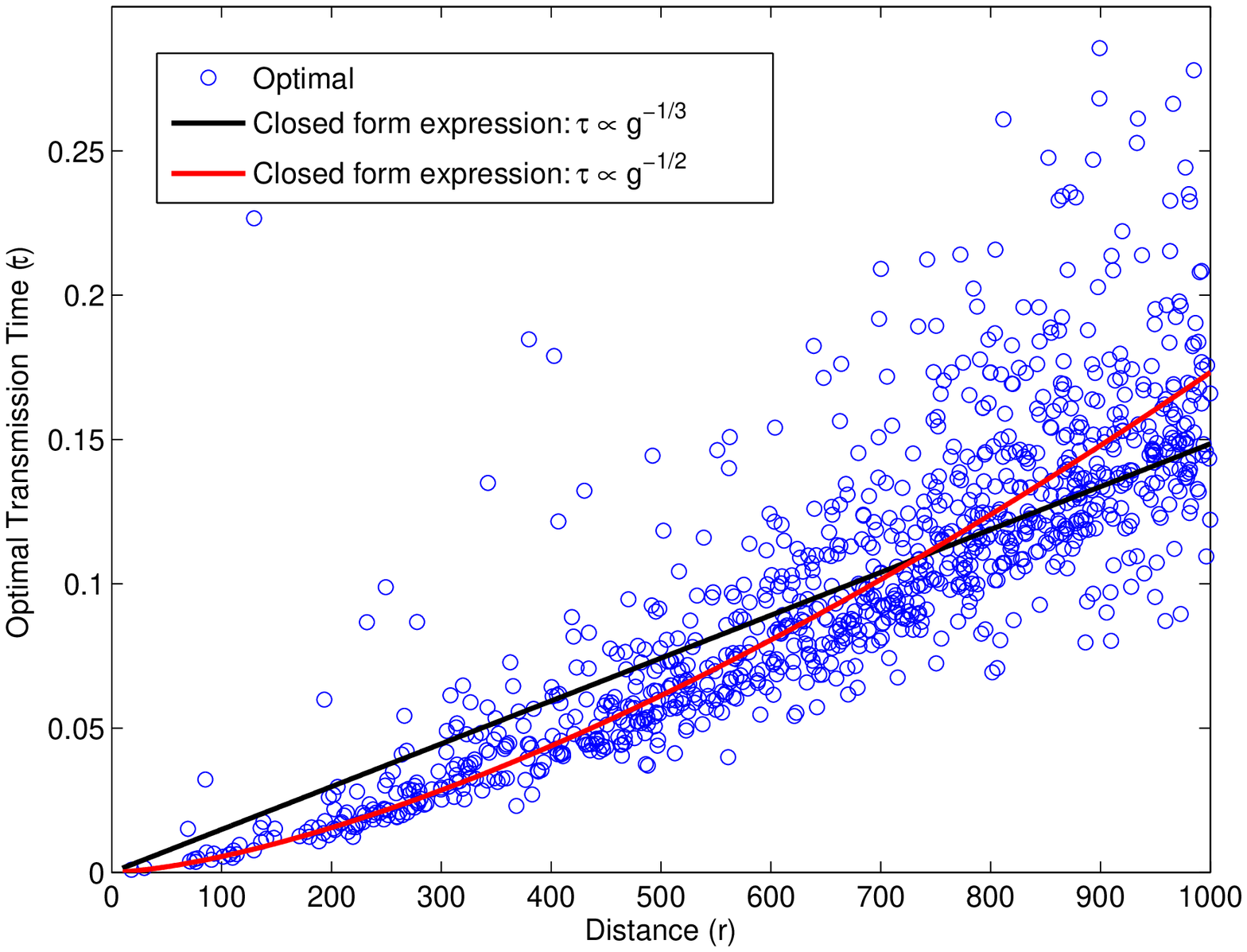}
\includegraphics[scale=.45]{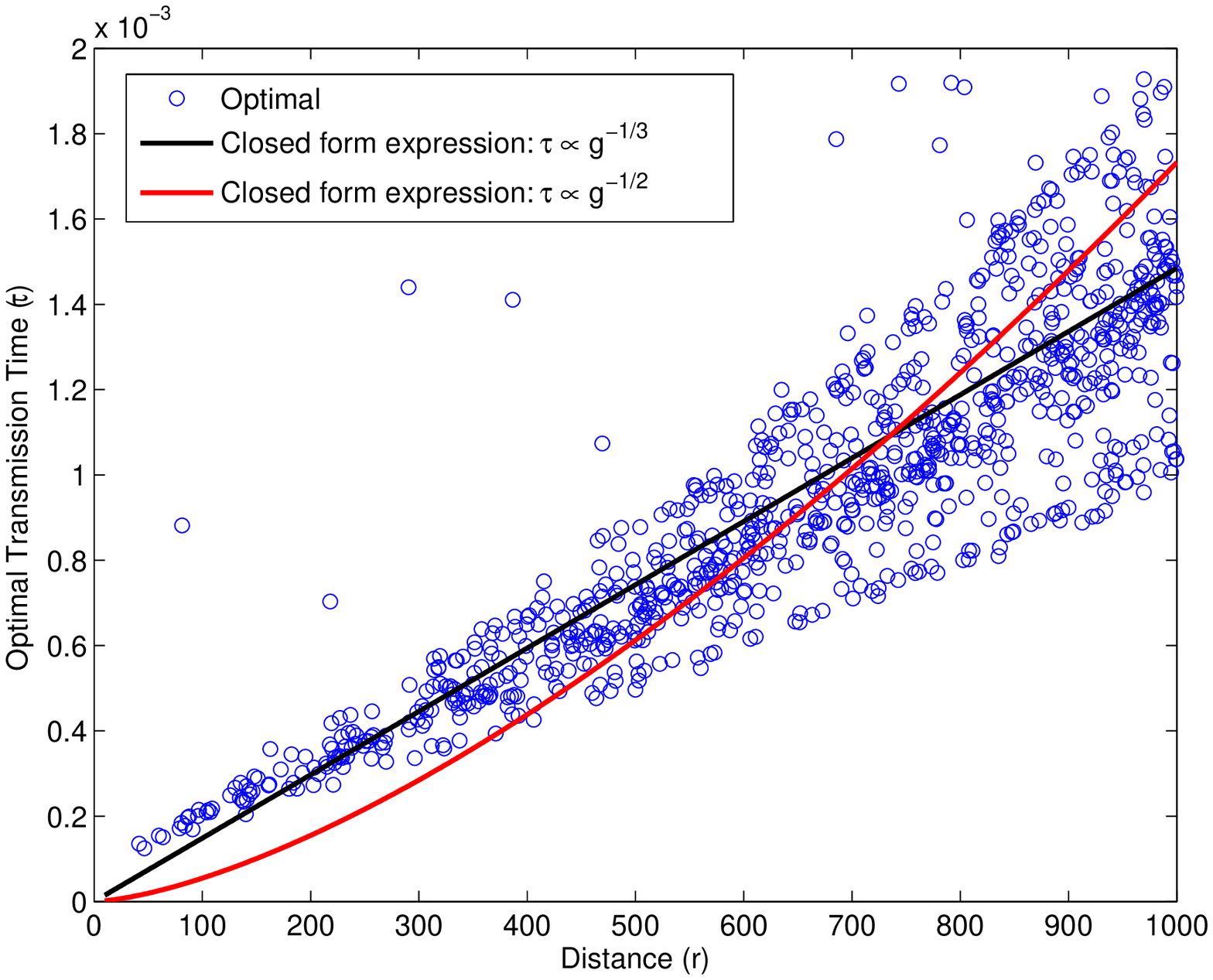}
\caption{Comparison of the closed form solution given by \eqref{eq:ClosedFormApprox} with the power optimal solution (scatter plot). (first) low load ($\lambda = 10$). (second) high load ($\lambda = 1000$). The fading and shadowing is averaged out.}
\label{fig:PowLambdaScatter}
\end{figure}

\begin{prop}[TDMA Outage]
The total system outage $\delta$ can be upper bounded in terms of the failure probabilities $\delta_1$ and $\epsilon_1$ as
\begin{align}
\delta \leq \epsilon_1 + \delta_1 (1-\epsilon_1).
\end{align}
\end{prop}

\begin{proof}
Using similar ideas as in Proposition~\ref{prop:CDMA_Pf}, the total outage $\delta$ can be written in terms of the failure event $\calA$ as:
\begin{align}
\delta =&\ \P[\calA | (1-\delta_1)N_a \geq K_{\rm max}] \P[(1-\delta_1)N_a \geq K_{\rm max}] + \nonumber \\
 &\ \P[\calA | (1-\delta_1)N_a < K_{\rm max}] \P[(1-\delta_1)N_a < K_{\rm max}]\\
 =&\ \epsilon_1 + \delta_1 (1-\epsilon_1),
\end{align}
where we bounded the term $\P[\calA | (1-\delta_1)N_a \geq K_{\rm max}]$ by $1$.
\end{proof}
This bound is tight especially when $\epsilon_1 \rightarrow 0$, which will  be the regime of interest for this paper.

\subsubsection{Near Optimal Closed-Form Solution} The form of the optimization problem is such that the exact closed form solutions are not possible. We now show that it is possible to derive approximate near-optimal closed form results, where the transmission time is solely a function of the channel gain of the device of interest independent of the channel gains of other devices. We demonstrate it for the power minimization problem, where the total transmit power can be expressed as:
\begin{align}
P &=  \sum\limits_{i=1}^{K} \frac{2^{\frac{L}{W \tau_i}}-1}{\mu g_i} \\
&= \frac{2^{\frac{L}{W \tau_1}}-1}{\mu g_1} + \frac{2^{\frac{L}{W \tau_2}}-1}{\mu g_2} \ldots \frac{ 2^{\frac{L}{W \left(\tau_s - \sum\limits_{j=1}^{K-1} \tau_j \right)}}-1}{\mu g_{K}}.
\end{align}
Minimizing the transmit power $P$ w.r.t. $\tau_1$ we get
\begin{align}
\label{eq:optcond}
\frac{\delta P }{\delta \tau_1} = 0 \Rightarrow \frac{2^{\frac{L}{W \tau_1}}}{g_1 \tau_1^2} = \frac{2^{\frac{L}{W \tau_N}}}{g_N \tau_N^2} \Rightarrow \tau_i^2 2^{\frac{-L}{W \tau_1}} \propto g_i^{-1}
\end{align}
\begin{remark} For small $\lambda$, $\tau$ is of the order of seconds, and for our choice of $L$ and $W$, $L/W = 10^{-3}$, which implies $2^{\frac{-L}{W \tau}} \rightarrow 1$. Therefore, $\tau \propto g^{-1/2}$. On the other hand, for large $\lambda$ (say $\lambda = 1000$), and $L/W = 10^{-3}$, we have $2^{\frac{L}{W \tau}} \approx \frac{2L}{W \tau}$. Therefore, $\tau \propto g^{-1/3}$ in this regime.
\end{remark}
From the above remark, we note that the transmission time can be expressed solely as a function of channel gain as
\begin{align}
\tau_i = \frac{f(g_i)}{\sum\limits_j f(g_j)} \tau_s,
\end{align}
which reduces further to
\begin{align}
\tau_i = \frac{f(g_i)}{\lambda E[f(g)]},
\end{align}
for reasonably high values of $\lambda$. Ignoring fading and shadowing and assuming $f(g_i) = g_i^{-1/n}$ that encompasses both the special cases of $n=2$ and $n=3$, it can be expressed in the closed form as follows:
\begin{align}
\tau_i = \frac{2+\gamma/n}{2\lambda} r^{\gamma/n} \frac{r_0^2 - r_i^2}{r_0^{2+\gamma/n} - r_i^{2+\gamma/n}}.
\label{eq:ClosedFormApprox}
\end{align}
In Fig.~\ref{fig:PowLambdaScatter}, we plot the closed form result \eqref{eq:ClosedFormApprox} along with a scatter plot of the optimal transmission times in two regimes: i) low load ($\lambda = 10$) and ii) high load ($\lambda = 1000$). We note that the approximations for both the regimes are quite accurate, which is more surprising for the low load case since the closed form expression was derived assuming that the value of $\lambda$ is reasonably high. Interested readers can refer to~\cite{DhiHuaC2012} for further details on the accuracy of this simple approach.

For the energy optimal solution, we present a similar comparison in Fig.~\ref{fig:EnLambdaScatter}, which shows that the approximation $\tau \propto g^{-1/3}$ is accurate both in the low and high load regimes. This leads to the same closed form solution in both the regimes.

\begin{figure}[ht]
\centering
\includegraphics[scale=.45]{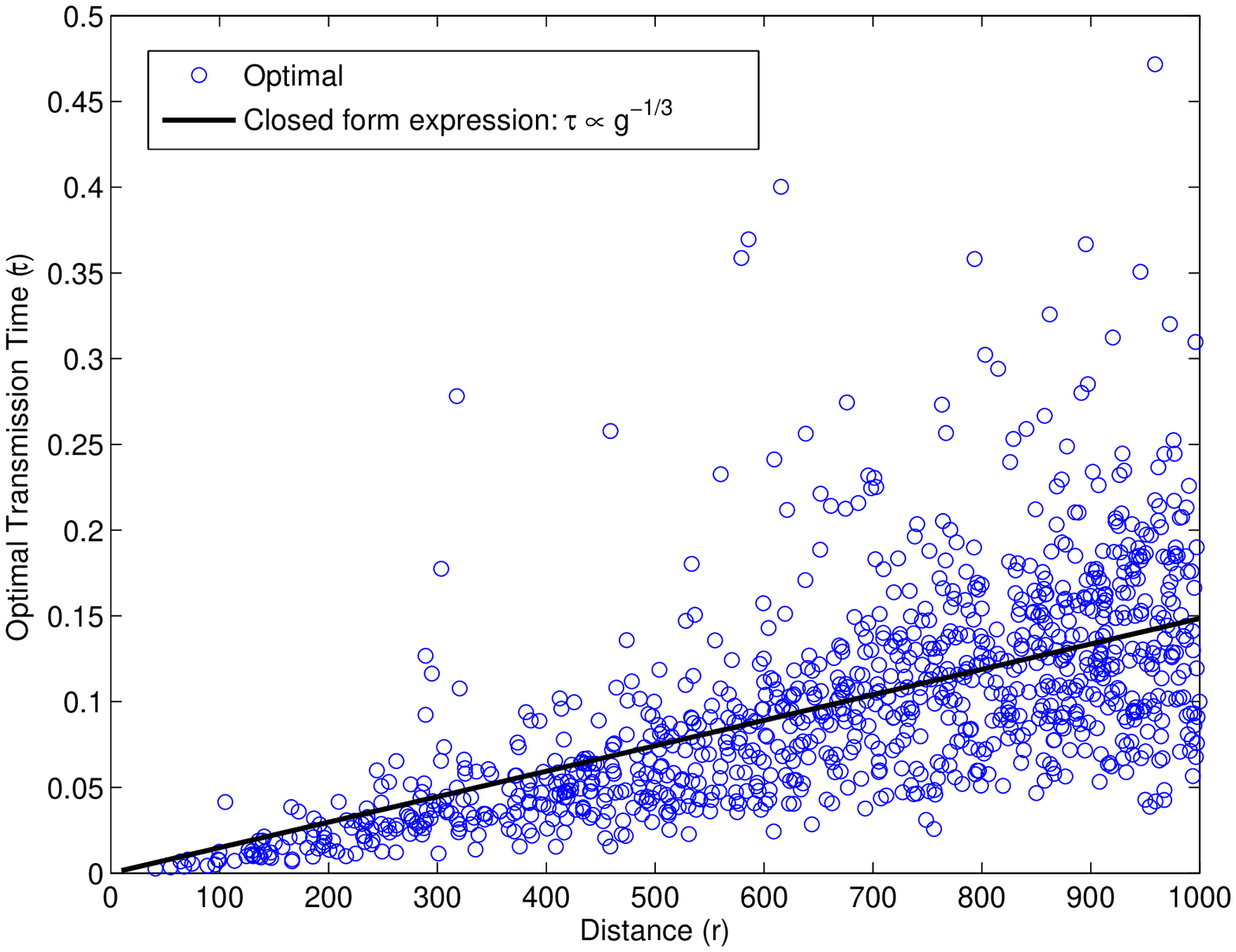}
\includegraphics[scale=.45]{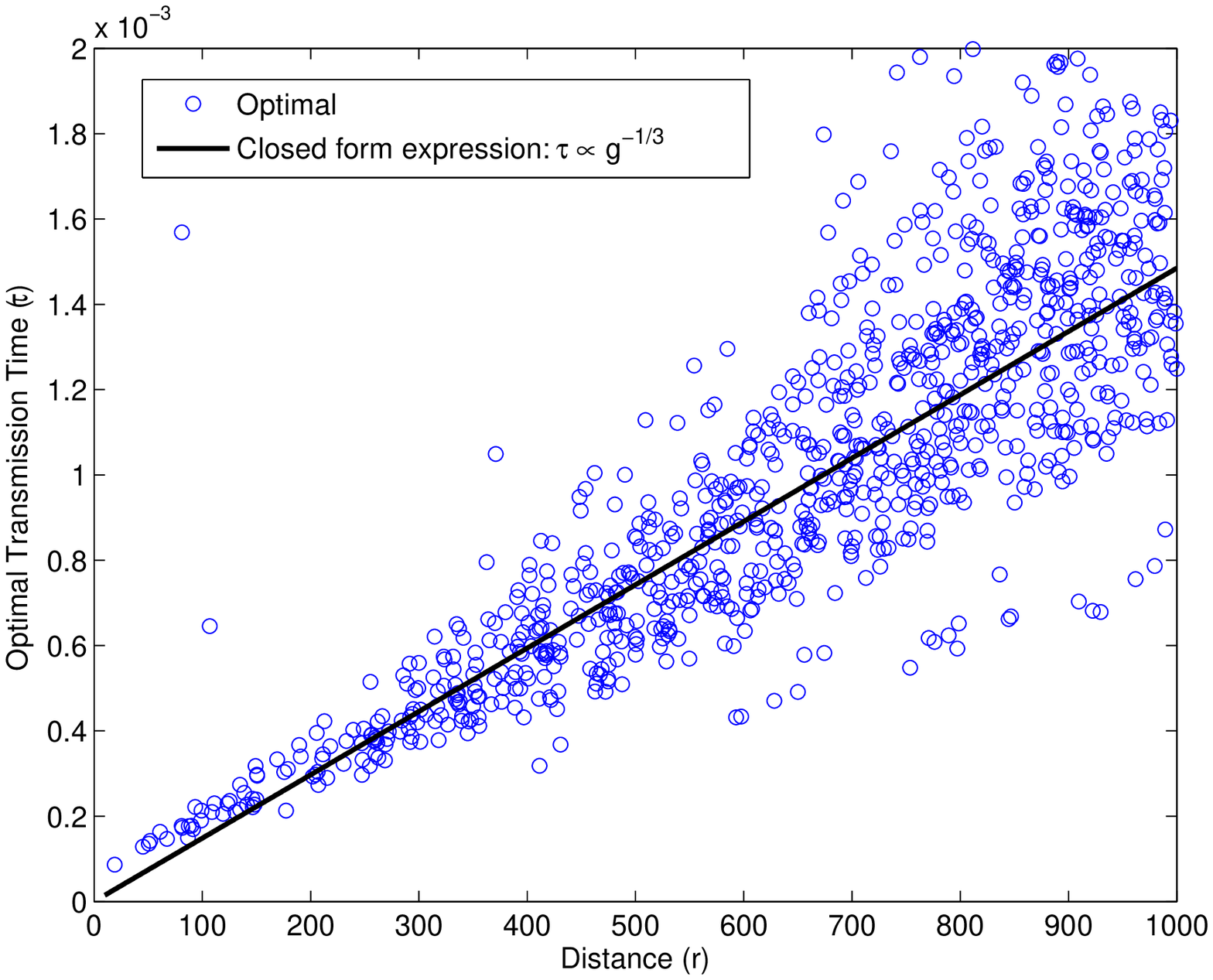}
\caption{Same as Fig.\ref{fig:PowLambdaScatter} except that the scatter point now corresponds to energy optimal solution. (first) low load ($\lambda = 10$). (second) high load ($\lambda = 1000$). The same closed form solution works in both the regimes.}
\label{fig:EnLambdaScatter}
\end{figure}

\subsubsection{Equal Time vs Power Optimal Allocation}
We now consider an even simpler case in which each device is allocated same transmission time and show that the transmit power under this simple strategy is always within a small constant of the optimal power in the parameter space of interest. The result is given by the following Theorem and the proof is given in Appendix~\ref{th:eqTDMA_proof}.

\begin{theorem}
\label{th:eqTDMA}
The ratio of total transmit powers of $K$ devices under uniform $(U_{g_1, g_2 \ldots g_K})$ and optimal $(P_{g_1, g_2 \ldots g_K})$ schedules is upper bounded by:
\begin{equation}
\frac{U_{g_1, g_2 \ldots g_K}}{P_{g_1, g_2 \ldots g_K}} \le \frac{2^{\frac{KL}{W\tau_s}}-1}{2^{\frac{L\sum\limits_{j=1}^K \sqrt{\frac{g_1}{g_j}}}{W\tau_s}}-1},
\label{eq:EqOptPowBound_Th}
\end{equation}
where $g_1 = \min\{g_j\}$.
\end{theorem}

This bound is surprisingly tight as shown in the following example.

\begin{example}[Bound in the parameter space of interest] For $K=1000$, $\tau_s = 1$ sec, $L=1000$ bits, $W = 1$ MHz and no fading, the bound given by \eqref{eq:EqOptPowBound_Th} is $\approx 2$, i.e., it is guaranteed that the transmit power under equal resource allocation is no more than around $3$ dB worse than the optimal power. In \cite{DhiHuaC2012} we have numerically shown that the actual gap is much smaller.
\label{eg:Bound_TDMAPow}
\end{example}

\subsubsection{Equal Time vs Energy Optimal Allocation}
On the same lines, we now compare the energy per bit required under equal time allocation with that of the energy optimal allocation. Before deriving the main result, we first derive an upper bound on the optimal time allocated to the device corresponding to the minimum channel gain. The proof is given in Appendix~\ref{lem:tau1_UBproof}.

\begin{lemma}
\label{lem:tau1_UB}
The optimal time allocated to the device corresponding to the minimum channel gain can be upper bounded as
\begin{align}
\tau_1^* \leq \frac{\tau_s}{\sum_{i=1}^K g_1/g_i},
\end{align}
where $g_1 = \min\{g_j\}$.
\end{lemma}

From this Lemma the following result follows. The proof is in Appendix~\ref{thm:TDMA_EnBoundproof}.

\begin{theorem} \label{thm:TDMA_EnBound}
The ratio of the total energy per bit under uniform $(U_{g_1, g_2 \ldots g_K})$ and optimal $(E_{g_1, g_2 \ldots g_K})$ schedules is upper bounded by:
\begin{equation}
\frac{U_{g_1, g_2 \ldots g_K}}{E_{g_1, g_2 \ldots g_K}} \leq \frac{\sum_{i=1}^K g_1/g_i}{K} \frac{2^\frac{KL}{W\tau_s}-1}{2^\frac{L\sum_{i=1}^K g_1/g_i}{W\tau_s}-1},
\label{eq:EqOptEnBound_Th}
\end{equation}
where $g_1 = \min\{g_j\}$.
\end{theorem}

\begin{cor}\label{cor:TDMAEn_limit}
In the limit of asymptotically low spectral efficiency, equal resource allocation is sum energy optimal on the space of TDMA strategies, i.e.,
\begin{align}
\lim_{W\tau_s \rightarrow \infty} \frac{U_{g_1, g_2 \ldots g_K}}{E_{g_1, g_2 \ldots g_K}} = 1.
\end{align}
\end{cor}

The proof of this corollary directly follows from Theorem~\ref{thm:TDMA_EnBound} under the limit $\frac{L}{W\tau_s} \rightarrow 0$ by using $\lim_{x\rightarrow 0} a^x - 1 = x\ln(a)$. In addition to this asymptotic result, the bound is surprisingly tight even in the parameter space of interest for M2M. This is shown in the following example.

\begin{example}
\label{eq:Bound_TDMAEn}
For the same system parameters as that of Example~\ref{eg:Bound_TDMAPow}, the ratio of energy per bit in the uniform and optimal allocation strategies is bounded above by $1.25$. Clearly, the value of optimization in terms of energy per bit is also limited for TDMA.
\end{example}
\subsection{FDMA System Design} The FDMA system design proceeds exactly in the same way as discussed for the TDMA case above. We will therefore highlight only the main differences in problem formulation. The goal here is to partition the available bandwidth $W$ into $K$ parts $\{W_1, W_2, \ldots, W_K\}$ so as to minimize the total transmit power or energy. The total power required in this case can be expressed as:
\begin{align}
P =  \sum\limits_{i=1}^{K} \frac{W_i}{W}\frac{2^{\frac{L}{W_i \tau_s}}-1}{\mu g_i },
\label{eq:FDMA_sumP}
\end{align}
and the total energy per bit can be expressed as:
\begin{align}
E_b = \frac{\tau_s }{L}\sum\limits_{i=1}^{K} \frac{W_i}{W}\frac{\left( 2^{\frac{L}{W_i \tau_s}}-1\right) }{\mu g_i }.
\end{align}
Unlike TDMA case, both the power and energy optimal schedules are exactly the same in case of FDMA. This is simply because the energy per bit can be expressed as a constant multiple of power, where the constant is independent of the optimization parameters. The constraint on the maximum transmit power translates to the minimum bandwidth required by each device depending upon its channel condition. This minimum bandwidth $W_{\rm min}$ is the solution of the following equation:
\begin{align}
\frac{L} {\tau_s W_{{\rm min}_i}} = \log_2 \left( 1 + P_{\rm max} \mu \left( \frac{W}{W_{{\rm min}_i}} \right) g_i  \right).
\end{align}
The energy or power minimization problem can now be formulated as:
\begin{align}
\min_{\{W_i\}}\ \ \ & \sum\limits_{i=1}^{K} v(W_i) \nonumber \\
s.t.\ \ \ & \sum\limits_{i=1}^{K} W_i \leq W \nonumber \\
& W_i \geq W_{min_i} \nonumber \\
& 1 \leq i \leq K, \nonumber
\end{align}
where $v(W_i)$ is the cost function representing power or energy per bit. We note that the form of the optimization problem is exactly the same as that of the TDMA problem discussed above in detail and hence most of the insights about the exact and approximate solutions carry over.

\begin{remark}[Feasibility and Maximum Load] As discussed for the TDMA counterpart in Remark~\ref{rem:TDMAfeasibility}, the optimization problem is feasible if the constraint $\sum_{i=1}^K W_{{\rm min}_i} \leq W $ is satisfied. Using this constraint, an approximate bound on $K$ can be derived as:
\begin{align}
K \leq  \frac{W}{\E[W_{\rm min}]}.
\end{align}
The maximum load can also be defined in the same way as done for the TDMA case.
\end{remark}
This completes the analysis of the coordinated strategies and we now compare the maximum load that can be handled by a base station using TDMA and FDMA in the following example.
\begin{example}[Maximum Load: FDMA vs TDMA]
\label{eg:TDMA_FDMA}
Choosing the same set of general system parameters as that of Example~\ref{eg:CDMA} and $\delta$ vanishingly small, the maximum load a base station can handle using TDMA and FDMA is $\approx 1200$ and $\approx 14700$ respectively. This clearly shows that it is optimal to partition over frequency. Although TDMA and FDMA are exactly the same from information theoretic sense, the difference in the optimal solution arises as a result of the peak power constraint that affects the two schemes differently as is apparent from the expressions of $\tau_{\rm min}$ and $W_{\rm min}$. We will comment more on this in Section~\ref{sec:numresults}.
\end{example}

As done for the TDMA case, we now compare equal resource allocation with the optimal allocation in terms of transmit power and energy per bit.

\subsubsection{Equal Bandwidth Allocation vs Optimal Allocation} Before deriving the main result, as in the TDMA case we first derive an upper bound on the optimal bandwidth allocated to the packet corresponding to the minimum channel gain. The proof follows on the same lines as that of Lemma~\ref{lem:tau1_UB} and a sketch is given in Appendix~\ref{lem:W1_UBproof} for completeness.

\begin{lemma}
\label{lem:W1_UB}
The optimal bandwidth allocated to the packet corresponding to the minimum channel gain can be upper bounded by
\begin{align}
W_1^* \leq \frac{W}{\sum_{i=1}^K g_1/g_i},
\end{align}
where $g_1 = \min\{g_j\}$.
\end{lemma}

From this Lemma, the following bound on the ratios of the total transmit powers follows. The same bound holds for the total energy per bit as well. The proof follows on the same lines as that of Theorem~\ref{thm:TDMA_EnBound} and  is hence skipped.

\begin{theorem}
The ratio of total transmit powers (and energies) under uniform and optimal schedules can be bounded as:
\begin{equation}
\frac{U_{g_1, g_2 \ldots g_K}}{P_{g_1, g_2 \ldots g_K}} \leq \frac{\sum_{i=1}^K g_1/g_i}{K} \frac{2^\frac{KL}{W\tau_s}-1}{2^\frac{L\sum_{i=1}^K g_1/g_i}{W\tau_s}-1},
\label{eq:EqOptPowBound_Th2}
\end{equation}
where $g_1 = \min\{g_j\}$.
\end{theorem}

\begin{remark}[FDMA Bound and Optimality] Note that this bound is the same as the one derived for the energy optimal solution of TDMA. Therefore, Example~\ref{eq:Bound_TDMAEn} is applicable in this case and hence the ratio of energy or power under uniform and optimal allocation strategies is upper bounded by $1.25$. Moreover, equal resource allocation is both sum power and sum energy optimal over the space of FDMA strategies. This follows from Corollary~\ref{cor:TDMAEn_limit}. 
\end{remark}

\subsubsection{Equal Bandwidth Allocation vs SIC} We conclude this section with an even stronger result, which proves that FDMA is sum power optimal over the space of general resource allocation strategies -- not limited to orthogonal -- in the limit of low spectral efficiency. The result follows by comparing the transmit power under equal bandwidth allocation with the global optimal transmit power achieved by SIC and is given by the following theorem. We denote the transmit power under equal bandwidth allocation by $U_{g_1, g_2 \ldots g_K}$ and by slight abuse of notation under SIC by $P_{g_1, g_2 \ldots g_K}$.

\begin{theorem}
Equal bandwidth allocation is sum power optimal over a space of general resource allocation strategies in the limit of low spectral efficiency, i.e., 
\begin{align}
\lim_{W\tau_s \rightarrow \infty}\frac{U_{g_1, g_2 \ldots g_K}}{P_{g_1, g_2 \ldots g_K}} = 1.
\end{align}
\end{theorem}

\begin{proof}
From the expressions of the sum power under SIC given by~\eqref{eq:SIC_sumP} and under FDMA given by~\eqref{eq:FDMA_sumP}, the ratio can be expressed as
\begin{align}
\frac{U_{g_1, g_2 \ldots g_K}}{P_{g_1, g_2 \ldots g_K}} = \frac{1}{K} \frac{2^\frac{LK}{W\tau_s} - 1}{2^\frac{L}{W\tau_s} - 1} \frac{\sum_{k=1}^K \frac{1}{g_k}}{\sum_{k=1}^K \frac{1}{g_k} 2^\frac{(k-1)L}{W\tau_s}},
\end{align}
from which the result follows by using $\lim_{x\rightarrow 0} a^x - 1 = x\ln(a)$.
\end{proof}
\section{Discussion and Numerical Results} \label{sec:numresults}

Note that the preferred choices for both the uncoordinated and coordinated strategies are clear from examples~\ref{eg:FDMA_RandAccess} and \ref{eg:TDMA_FDMA}. In case of uncoordinated strategies, example~\ref{eg:FDMA_RandAccess} shows that the random access CDMA supports order of magnitude higher load than random access FDMA and in case of coordinated strategies, example~\ref{eg:TDMA_FDMA} shows that coordinated FDMA is clearly a better choice over coordinated TDMA. Therefore, in this section we will not consider random access FDMA and coordinated TDMA, except when we compare energy optimal solutions. Interested readers can refer to~\cite{DhiHuaC2012} for a detailed discussion on the optimal power TDMA results. Instead, our main focus in this section will be on random access CDMA, coordinated FDMA and SIC in various regimes of interest.


\subsection{Optimal Transmit Power}
We first compare the optimal transmit power required for random access CDMA, coordinated FDMA and SIC in Fig.~\ref{fig:Pow_MainComparison}, where we plot both the mean power and the $95^{th}$ percentile of the power. In both the cases, we observe that the performance of random access CDMA is close to that of FDMA in low and moderate arrival rates, especially if one accounts for the signaling overhead required in FDMA case. This shows that CDMA random access may be preferred at low to moderate arrival rates due to the apparent simplicity of the resulting system design. As expected, FDMA is the only option at high arrival rates and its performance is very close to that of the optimal SIC performance. One possible system design for this regime is to first use random access CDMA to establish uplink connection by sending small payload containing only the control information. The load supported by random access CDMA in this case will be very high because of the small payload size. The base station can then decide about the frequency allocation and relay this information back to the devices through broadcast signals. Recall that these two system designs, corresponding to high and low loading, are exactly the same as the one stage and two stage design examples discussed in Section~\ref{sec:sysmod}.

\subsection{Optimal Energy per Bit}
We now compare the three strategies along with coordinated TDMA in terms of the energy per bit in Fig.~\ref{fig:AvgEb_MainComparison}. We first note that the performance of TDMA and FDMA cases is similar. This is intuitive because both TDMA and FDMA are exactly the same from information theoretic sense. The difference in the transmit powers required in both the cases is a result of the peak power constraint that affects the two schemes differently. Now comparing the three main candidate strategies, we again note that random access CDMA performance is close to that of coordinated FDMA at low to moderate loading, and coordinated FDMA is the only choice at high loading. The SIC performance is close to that of FDMA in this case as well. These observations lead to the same design guidelines as discussed above in case of optimal transmit power.

\subsection{Coordinated FDMA vs SIC at very High Loading}
In Fig.~\ref{fig:Pow_FDMA_1_13001}, we consider very high loading regime and restrict our comparison to SIC and coordinated FDMA for which we consider both optimal and equal bandwidth allocation strategies. We evaluate both the mean optimal power and $95^{\rm th}$ percentile of the power for both the strategies. Comparing equal bandwidth allocation with optimal FDMA, we note that equal resource allocation is near-optimal even at very high loading, which is consistent with our analysis. On the other hand, we observe comparatively higher gap between the two strategies in terms of $95^{\rm th}$ percentile of the power. The gap is, however, not too significant since even at $\lambda = 6000$ the transmit powers differ by less than $3$ dB. 
On the other hand, the optimal SIC performance is comparatively much better, with the performance gap being more than $6-7$ dB from the equal bandwidth allocation and around $4$ dB from optimal bandwidth allocation at $\lambda = 6000$. However, it is important to note that optimal SIC requires perfect channel estimates for all the devices, which may not be realistic especially at high arrival rates. Therefore, the performance gap, even at very high loading, is unlikely to be significant if we account for high implementation losses expected in SIC. Exact quantification of these losses is out of the scope of this paper.

\begin{figure}[t]
\centering
\includegraphics[width=\columnwidth]{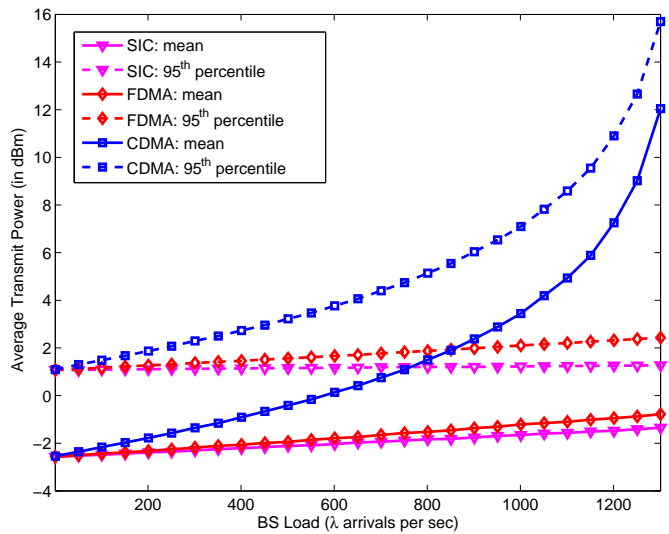}
\caption{Comparison of power optimal solutions of SIC and FDMA with random access CDMA.}
\label{fig:Pow_MainComparison}
\end{figure}

\begin{figure}[t]
\centering
\includegraphics[width = \columnwidth]{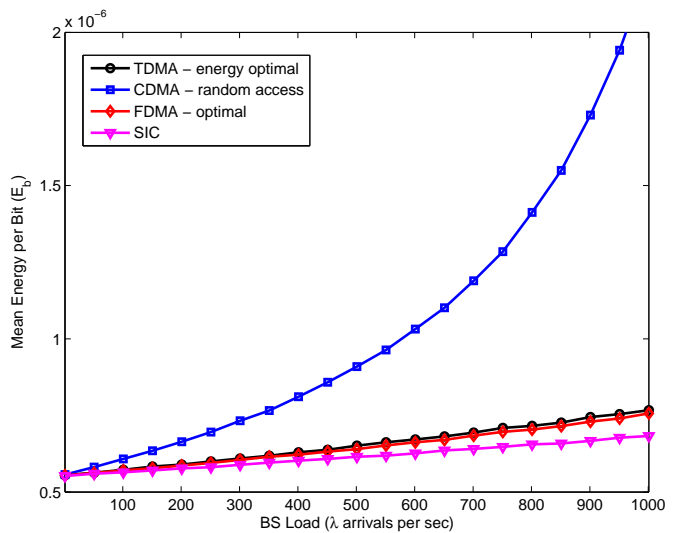}
\caption{Comparison of energy optimal solution of TDMA and FDMA with random access CDMA.}
\label{fig:AvgEb_MainComparison}
\end{figure}

\begin{figure}[t]
\centering
\includegraphics[width=\columnwidth]{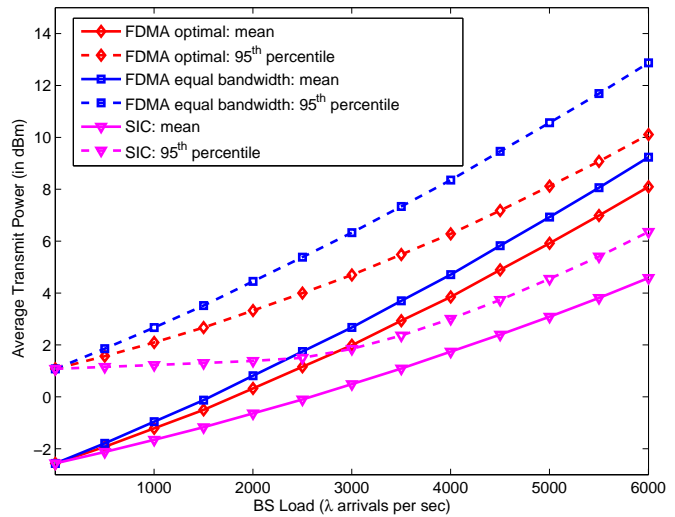}
\caption{Comparison of the optimal power  and equal bandwidth allocation solutions for FDMA. 
}
\label{fig:Pow_FDMA_1_13001}
\end{figure}

\section{Conclusions}
In this paper, we have developed a systematic framework to study the power and energy optimal system design in the parameter space of interest for M2M communications. For comparison, we consider a variety of uncoordinated strategies, such as random access CDMA and FDMA, and coordinated strategies, such as SIC, FDMA and TDMA. While the coordinated FDMA is the best practical strategy overall, random access CDMA is almost as good when the base station is lightly loaded and is a strong candidate for the actual system design due to its apparent simplicity and no signaling over head. Additionally, we have shown that the value of optimization is in general small both for TDMA and FDMA, and simpler resource allocation strategies, such as equal resource allocation achieve performance within a small constant of the optimal performance. In the limit of asymptotically small spectral efficiency, equal bandwidth allocation in coordinated FDMA is shown to be power optimal over the space of general resource allocation strategies. An important extension of this work includes accounting for the implementation losses and signaling overhead while comparing various multiple access strategies.

\appendices
\section{Proof of Theorem~\ref{th:eqTDMA}} \label{th:eqTDMA_proof}
Let $U_{g_1, g_2 \ldots g_K}$ denote the total transmit power under uniform schedule (equal transmission time) and $P_{g_1, g_2 \ldots g_K}$ denote the power under the optimal schedule. Further let $U_{g_i, \ldots g_j}$ and $P_{g_i, \ldots g_j}$ be the transmit powers of the subset of users under uniform and optimal schedules, respectively. Without loss of generality, we assume that the channel gains $g_i$ are indexed in the increasing orders of magnitude. Let the total transmission time of $K$ devices be $\tau$, with $\tau_j^*$ representing the optimal transmission time of $j^{th}$ packet. The ratio of powers under two schemes can now be expressed as:
\begin{align}
\frac{U_{g_1, g_2 \ldots g_K}}{P_{g_1, g_2 \ldots g_K}} &= \frac{U_{g_1} + U_{g_2} + \ldots U_{g_K} }{P_{g_1} + P_{g_2} + \ldots P_{g_K} }\\
& \leq \max \left\{ \frac{U_{g_1}}{P_{g_1}}, \frac{U_{g_2}}{P_{g_2}}, \ldots, \frac{U_{g_K}}{P_{g_K}}   \right\}.
\label{eq:Max_Kterms}
\end{align}
The ratio $U_{g_k}/P_{g_k}$ can be expressed as:
\begin{equation}
\frac{U_{g_k}}{P_{g_k}} = \frac{2^\frac{KL}{W\tau_s}-1}{2^\frac{L}{W\tau_k^*}-1}.
\end{equation}
From the following optimality condition derived in \eqref{eq:optcond}:
\begin{equation}
\frac{2^{L/W\tau_k^*} } {g_k \tau_k^{*^2}} = \frac{2^{L/W\tau_1^*} } {g_1 \tau_1^{*^2}},
\label{eq:OptCond}
\end{equation}
we note that $\tau_1^* \geq \tau_k^*$ $\forall\ k$ since $g_1 \leq g_k$. Therefore, \eqref{eq:Max_Kterms} can be expressed as:
\begin{equation}
\frac{U_{g_1, g_2 \ldots g_K}}{P_{g_1, g_2 \ldots g_K}} \leq \frac{U_{g_1}}{P_{g_1}}.
\end{equation}
To bound $\frac{U_{g_1}}{P_{g_1}}$, we first derive the following inequality from the optimality condition:
\begin{equation}
\tau_k^* \geq \tau_1^* \sqrt{\frac{g_1}{g_k}}, \forall\ k.
\label{eq:tauk_LB}
\end{equation}
Using \eqref{eq:tauk_LB}, we now derive an upper bound on $\tau_1^*$ as follows:
\begin{align}
\tau_s = \sum\limits_{j=1}^K \tau_j^* \geq \tau_1^* \sum\limits_{j=1}^K \sqrt{\frac{g_1}{g_j}} \Rightarrow \tau_1^* \leq \tau_s/ \sum\limits_{j=1}^K \sqrt{\frac{g_1}{g_j}}.
\label{eq:tau1_UB}
\end{align}
Using \eqref{eq:tau1_UB}, the ratio $U_{g_1}/P_{g_1}$ can be bounded as:
\begin{equation}
\frac{U_{g_1}}{P_{g_1}} = \frac{2^{KL/W\tau_s}-1}{2^{L/W\tau_1^*}-1} \leq \frac{2^{\frac{KL}{W\tau_s}}-1}{2^{\frac{L\sum\limits_{j=1}^K \sqrt{\frac{g_1}{g_j}}}{W\tau_s}}-1},
\end{equation}
which completes the proof. \hfill \IEEEQEDclosed

\section{Proof of Lemma~\ref{lem:tau1_UB}} \label{lem:tau1_UBproof}
Assume that the channel gains $g_i$ are indexed in the increasing orders of magnitude. For notational simplicity, denote the transmit energy per bit of the $i^{th}$ device transmitting for time $0<x\leq\tau_s$ by $\Psi_i(x)$:
\begin{align}
\Psi_i(x) = \frac{x}{L} \frac{2^{\frac{L}{x W}}-1}{\mu g_i}.
\end{align}
Note that $\Psi_i(x)$ is a monotonically decreasing function in $x$. Therefore,
\begin{align}
\Psi_i'(x) &= \frac{\delta}{\delta x}\Psi_i (x)\\
 &= \frac{1}{L \mu g_i}\left(2^{\frac{L}{x W}} - 1 - 2^{\frac{L}{x W}} \frac{L}{x W} \ln 2 \right) \leq 0,\ x > 0.
\label{eq:neg_slope1}
\end{align}
Now we derive the optimality condition for the energy optimal TDMA allocation. The total transmit energy per bit can be expressed as:
\begin{align}
E &=  \sum\limits_{i=1}^{K} \Psi_i(\tau_i) = \Psi_1(\tau_1) + \Psi_2(\tau_2) + \ldots
\Psi_K\left( \tau_s - \sum_{i=1}^{K-1} \tau_i \right).
\end{align}
Minimizing the transmit energy $E$ w.r.t. $\tau_1$ we get
\begin{align}
\label{eq:optcond1}
\frac{\delta E }{\delta \tau_1} = 0 \Rightarrow \Psi_1'(\tau_1) = \Psi_j'(\tau_j)\ \forall\ j,
\end{align}
which has to be satisfied for the optimal transmission time $\tau_i^*$ as well. Therefore,
\begin{align}
\frac{g_1}{g_i} &= \frac{ 2^{\frac{L}{\tau_1^* W}} \frac{L}{\tau_1^* W} \ln 2  + 1 - 2^{\frac{L}{\tau_1^* W}}}{2^{\frac{L}{\tau_i^* W}} \frac{L}{\tau_i^* W} \ln 2  + 1 - 2^{\frac{L}{\tau_i^* W}} }\\
&= \frac{\frac{L}{\tau_1^* W}}{\frac{L}{\tau_i^* W}}  \left(\frac{2^{\frac{L}{\tau_1^* W}} \ln 2 - \frac{\tau_1^* W}{L} (2^{\frac{L}{\tau_1^* W}}-1)}{2^{\frac{L}{\tau_i^* W}} \ln 2 - \frac{\tau_i^* W}{L} (2^{\frac{L}{\tau_i^* W}}-1)} \right) \\
&\stackrel{(a)}{\leq} \frac{\tau_i^*}{\tau_1^*},
\end{align}
where $(a)$ follows from (i) the function $\Phi(x) = 2^\frac{1}{x}\ln 2 - x (2^\frac{1}{x} - 1)$ is a decreasing function of $x$ and hence the ratio $\Phi(x_1)/\Phi(x_2) \leq 1$ for $x_1 \geq x_2$, and (ii) $\tau_1^* \geq \tau_i^*$, $\forall\ i$ which follows from the optimality condition along with the fact that $\Psi_i'(x)$ is a monotonic function of $x$.
Note that
\begin{align}
\tau_s = \sum_i \tau_i^* \geq \tau_1^*  \sum_{i=1}^K \frac{g_1}{g_i} \Rightarrow \tau_1^* \leq \frac{\tau_s}{\sum_{i=1}^K g_1/g_i},
\end{align}
which completes the proof.\hfill \IEEEQEDclosed

\section{Proof of Theorem~\ref{thm:TDMA_EnBound}} \label{thm:TDMA_EnBoundproof}
Let $U_{g_1, g_2 \ldots g_K}$ denote the total energy per bit under uniform schedule (equal transmission time) and $E_{g_1, g_2 \ldots g_K}$ denote the energy per bit under the optimal schedule. Further let $U_{g_i, \ldots g_j}$ and $P_{g_i, \ldots g_j}$ be the transmit powers of the subset of users under uniform and optimal schedules, respectively. Rest of the setup remains the same as that of the proof of Theorem~\ref{th:eqTDMA}. 
The ratio of energies under two scheduling schemes can now be expressed as:
\begin{align}
\frac{U_{g_1, g_2 \ldots g_K}}{E_{g_1, g_2 \ldots g_K}} &= \frac{U_{g_1} + U_{g_2} + \ldots U_{g_K} }{E_{g_1} + E_{g_2} + \ldots E_{g_K} }\\
& \leq \max \left\{ \frac{U_{g_1}}{E_{g_1}}, \frac{U_{g_2}}{E_{g_2}}, \ldots, \frac{U_{g_K}}{E_{g_K}}   \right\}.
\label{eq:Max_Kterms1}
\end{align}
The ratio $U_{g_k}/E_{g_k}$ can be expressed as:
\begin{equation}
\frac{U_{g_k}}{E_{g_k}} = \frac{\tau_s}{K\tau_k^*} \frac{2^\frac{KL}{W\tau_s}-1}{2^\frac{L}{\tau_k^*W}-1}.
\end{equation}
From the optimality condition derived in Lemma~\ref{lem:tau1_UB}, $\tau_1^* \geq \tau_i^*$ $\forall\ i$ since $g_1 \leq g_k$. Therefore, \eqref{eq:Max_Kterms1} can be expressed as:
\begin{equation}
\frac{U_{g_1, g_2 \ldots g_K}}{E_{g_1, g_2 \ldots g_K}} \leq \frac{U_{g_1}}{E_{g_1}}.
\end{equation}
As stated in the proof of Lemma~\ref{lem:tau1_UB}, $E_{g_1}$ is a monotonically decreasing function of the transmission time. Hence the upper bound on the ratio follows from the result of Lemma~\ref{lem:tau1_UB}. \hfill \IEEEQEDclosed

\section{Proof of Lemma~\ref{lem:W1_UB}} \label{lem:W1_UBproof}
Denote the transmit powers of the subset of the users under equal bandwidth allocation and optimal allocation by $U_{g_i, \ldots g_j}$ and $P_{g_i, \ldots g_j}$, respectively. Without loss of generality, assume that the channel gains $g_i$ are indexed in the increasing orders of magnitude. Let the total bandwidth to be allocated to $K$ packets is $W$, with $W_j^*$ representing the optimal bandwidth of $j^{th}$ packet. For notational simplicity, denote the transmit power of the $i^{th}$ user using bandwidth $0<x<W$ by $\Psi_i(x)$:
\begin{align}
\Psi_i(x) = \left( \frac{x}{W} \right) \frac{2^{\frac{L}{x \tau}}-1}{\mu g_i}.
\end{align}
Note that the functional form of $\Psi_i$ is the same as that in Lemma~\ref{lem:tau1_UB}. Therefore, the proof essentially follows on the same lines. Expressing the total transmit power as:
\begin{align}
P &=  \sum\limits_{i=1}^{K} \Psi_i(W_i)\\
 &= \Psi_1(W_1) + \Psi_2(W_2) + \ldots
\Psi_K\left( W - \sum_{i=1}^{K-1} W_i \right),
\end{align}
and following the same methodology as of Lemma~\ref{lem:tau1_UB}, we can derive the following lower bound on the ratio of channel gains:
\begin{align}
\frac{g_1}{g_i} \leq \frac{W_i^*}{W_1^*}.
\end{align}
Now note that
\begin{align}
W = \sum_i W_i^* \geq W_1^*  \sum_{i=1}^K \frac{g_1}{g_i} \Rightarrow W_1^* \leq \frac{W}{\sum_{i=1}^K g_1/g_i},
\end{align} 
which completes the proof. \hfill \IEEEQEDclosed
\bibliographystyle{IEEEtran}
\bibliography{M2M_Journal1_v1.7}

\end{document}

%% file: notation.tex
\def\nb0{{\mathbf{0}}}
\def\nb1{{\mathbf{1}}}

\def\ncalX{{\mathcal{X}}}

\def\pois{{\rm Pois}}

\newtheorem{lemma}{Lemma}

\newtheorem{ndef}{Definition}

\newtheorem{theorem}{Theorem}
\newtheorem{prop}{Proposition}
\newtheorem{cor}{Corollary}
\newtheorem{example}{Example}
\newtheorem{remark}{Remark}

\def\E{\mathbb{E}}

\def\P{\mathbb{P}}
\def\pc{\mathtt{P_c}}

\def\sinr{\mathtt{SINR}}			
\def\snr{\mathtt{SNR}}

\def\calN{\mathcal{N}}

\def\calA{\mathcal{A}}

